%% file: Main.tex
\documentclass[submission,copyright]{eptcs}
\providecommand{\event}{FICS 2012} 

\usepackage[UKenglish]{babel}
\usepackage{microtype}
\usepackage[toc,page]{appendix}
\usepackage[T1]{fontenc}
\usepackage{amsfonts,amsmath,amssymb,amsthm}
\usepackage{tikz} 
	\usetikzlibrary{shapes}
\usepackage{stmaryrd}
\usepackage[normalem]{ulem} 
\usepackage{calligra}

\renewcommand{\epsilon}{\varepsilon}

\newcommand{\termes}{\mathcal{T}}

\newcommand{\barre}[1]{\overline{#1} }

\newcommand{\underAccolade}[2]{\underset{#1}{\underbrace{#2}} }
\newcommand{\variables}{\mathcal{V}}

\newcommand{\substi}[3]{{#3}_{[#1 \mapsto #2]}}
\newcommand{\order}{\text{order}}
\newcommand{\Acc}{\text{Acc}}
\renewcommand{\a}{\textnormal{\texttt{a}}}
\renewcommand{\b}{\textnormal{\texttt{b}}}
\renewcommand{\c}{\textnormal{\texttt{c}}}
\newcommand{\redex}{\text{redex} }
\newcommand{\termi}{\Sigma}
\newcommand{\nontermi}{\mathcal{N}}

\newcommand{\rewrite}{\mathcal{R}}
\newcommand{\valtree}[1]{\| #1 \| }
\newcommand{\nbvar}[1]{\lceil #1 \rceil }
\newcommand{\sem}[1]{\llbracket #1 \rrbracket }

\newcommand{\defin}[1]{\textit{\textbf{#1}}}

\newtheorem{theorem}{Theorem}
\newtheorem{proposition}[theorem]{Proposition}
\newtheorem{lemma}[theorem]{Lemma}

\title{IO vs OI in Higher-Order Recursion Schemes}
\author{Axel Haddad 
\institute{LIAFA (Universit\'e Paris 7 \& CNRS) \quad \& \quad LIGM (Universit\'e Paris Est \& CNRS)}
}
\def\titlerunning{IO vs OI in Higher-Order Recursion Schemes}
\def\authorrunning{A. Haddad}
\begin{document}
\maketitle

\begin{abstract}
We propose a study of the modes of derivation of higher-order recursion schemes, proving that value trees obtained from schemes using innermost-outermost derivations (IO) are the same as those obtained using unrestricted derivations. 

Given that higher-order recursion schemes can be used as a model of functional programs, innermost-outermost derivations policy represents a theoretical view point of call by value evaluation strategy.

\end{abstract}

\section{Introduction}
	\input{Introduction}
\section{Preliminaries}
	\input{Preliminaries}
\section{From OI to IO}\label{OI2IO}
	\input{OI2IOShort}
\section{From IO to OI}\label{IO2OI}
	\input{IntroIO2OI}
	\input{kobaIOShort}

	\input{selfcorrectingShort}

\section{Conclusion}
	\input{Conclusion}

\bibliographystyle{eptcs}
\bibliography{biblio}

\newpage

%

\end{document}

%% file: Introduction.tex
Recursion schemes have been  first considered as a model of computation,
representing the syntactical aspect of a recursive program \cite{Nivat72,Courcelle78a,Courcelle78b,CN78} . At first, (order-1) schemes were modelling simple recursive programs whose
functions only take values as input (and not functions). Since, higher-order 
versions of recursion schemes  \cite{Indermark76,Damm77a,Damm77b,Damm82,ES77,ES78} have been studied. 

More recently, recursion schemes were studied as generators of infinite ranked trees and the focus was on deciding logical properties of those trees \cite{KNU02,ES77,HMOS08,Aehlig06,Kobayashi09,KO09}.

As for programming languages, the question of the
evaluation policy has been widely studied. Indeed, different policies results in the different evaluation \cite{ES77, ES78,Damm82}.  There are two main evaluations
policy for schemes: outermost-innermost
derivations (OI) and inner-outermost IO derivations, respectively corresponding to call by need
and call by value in programming languages.

Standardization theorem for the lambda-calculus shows that for any
scheme, outermost-innermost
derivations (OI) lead to the same tree as unrestricted derivation. However, this
is not the case for IO derivations. In this paper we prove that the situation is different for schemes. Indeed, we establish that the trees produced using schemes with IO policy are the same as those produced
using schemes with OI policy. 
For a given a scheme of order $n$, we can use a simplified continuation passing style transformation, to get a new scheme of order $n+1$ in which IO derivations will be the same as OI derivations in the initial scheme (Section~\ref{OI2IO}). 
Conversely, in order to turn a scheme into another one in which unrestricted
derivations lead to the same tree as IO derivations in the initial scheme,
we adapt Kobayashi's~\cite{Kobayashi09} recent results on HORS model-checking, to compute some
key properties over terms (Section~\ref{koba}). Then we embed these properties into
a scheme turning it into a self-correcting scheme of the same order of the initial scheme, in which OI and IO
derivations produce the same tree (Section~\ref{selfcorrecting}).

%% file: Preliminaries.tex
\defin{Types} are defined by the grammar  $\tau ::= o \ |\  \tau \rightarrow \tau$;  $o$ is called the \defin{ground type}. Considering that $\rightarrow$ is associative to the right (i.e. $\tau_1 \rightarrow (\tau_2 \rightarrow \tau_3)$ can be written $\tau_1 \rightarrow \tau_2 \rightarrow \tau_3$), any type $\tau$ can be written uniquely as $\tau_1 \rightarrow ... \rightarrow \tau_k \rightarrow o$. The integer $k$ is called the \defin{arity} of $\tau$. We define the \defin{order of a type} by $\order(o) = 0$ and $\order (\tau_1 \rightarrow \tau_2) = \max (\order(\tau_1)+1,\order(\tau_2))$. For instance  $o \rightarrow o \rightarrow o \rightarrow o$ is a type of order 1 and arity 3, $(o\rightarrow o) \rightarrow ( o \rightarrow o )$, that can also be written $(o\rightarrow o) \rightarrow o \rightarrow o$  is a type of order $2$. Let $\tau^\ell\rightarrow \tau'$ be a shortcut for $ \underAccolade{\ell \ times}{\tau \rightarrow ... \rightarrow \tau} \rightarrow \tau'$.

Let $\Gamma$ be a finite set of symbols such that to each symbol is associated a type. Let $\Gamma^\tau$ denote the set of symbols of type $\tau$. For all type $\tau$, we define the set of \defin{terms} of type $\termes^\tau(\Gamma)$ as the smallest set satisfying:  $\Gamma^\tau \subseteq \termes^\tau(\Gamma)$ and $ \bigcup_{\tau'} \{ t \ s \ | \ t\in\termes^{\tau' \rightarrow \tau} (\Gamma), s \in\termes^{\tau'} (\Gamma)\} \subseteq  \termes^\tau(\Gamma)$. If a term $t$ is in $\termes^\tau(\Gamma)$, we say that $t$ has type $\tau$. We shall write $\termes(\Gamma)$ as the set of terms of any type, and $t:\tau$ if $t$ has type $\tau$. The arity of a term $t$, $arity(t)$, is the arity of its type.
Remark that any term $t$ can be uniquely written as $t=\alpha \ t_1...t_k$ with $\alpha\in \Gamma$. We say that $\alpha$ is the \defin{head} of the term $t$. For instance, let $\Gamma = \{ F : (o\rightarrow o) \rightarrow o \rightarrow o \ ,\ G : o\rightarrow o \rightarrow o \ , \ H : (o \rightarrow o) \ , \ \a: o \}$: $F\ H$ and $G \ \a$ are terms of type $o \rightarrow o$;  $F (G \ \a) \ (H \ (H \ a))$ is a term of type $o$; $F \ \a$ is not a term since $F$ is expecting a first argument of type $o \rightarrow o$ while $\a$ has type $o$.\medskip


Let $t:\tau$, $t':\tau'$ be two terms, $x:\tau'$ be a symbol of type $\tau'$, then we write $\substi{x}{t'}{t}:\tau$ the term obtained by substituting all occurences of $x$ by $t'$ in the term $t$. A \defin{$\tau$-context} is a term $C[\bullet^\tau] \in \termes(\Gamma \uplus \{ \bullet^\tau:\tau  \} )$ containing exactly one occurrence of $\bullet^\tau$; it can be seen as an application turning a term into another, such that for all $t:\tau$, $C[t] = \substi{\bullet^\tau}{t}{C[\bullet^\tau]}$. In general we will only talk about ground  type context where $\tau=o$ and we will omit to specify the type when it is clear. For instance, if $C[\bullet] = F \bullet \ (H \ (H \ a))$ and $t'= G\ \a$ then $C[t'] = F \ (G \ \a) \ (H\ (H \ \a))$.\medskip


Let $\termi$ be a set of symbols of order at most $1$ (i.e. each symbols has type $o$ or $o \rightarrow ... \rightarrow o$) and $\bot:o$ be a fresh symbol. A \defin{tree} $t$ over $\termi \uplus{\bot}$ is a mapping $t : dom^t \rightarrow \termi \uplus{\bot}$, where $dom^t$ is a prefix-closed subset of $\{1,...,m\}^*$ such that if $u\in dom^t$ and $t(u)=a$ then $\{ j \ | \ uj \in dom^t\} = \{ 1,...,arity(a)\}$. Note that there is a direct bijection between ground terms of $\termes^o(\termi \uplus{\bot})$ and finite trees . Hence we will freely allow ourselves to treat ground terms over $\termi \uplus{\bot}$ as trees. We define the partial order $\sqsubseteq$ over trees as the smallest relation satisfying $\bot \sqsubseteq t$ and $t\sqsubseteq t$ for any tree $t$, and $a \ t_1 ... t_k \sqsubseteq a \ t'_1 ... t'_k$ iff $t_i \sqsubseteq t'_i$. Given a (possibly infinite) sequence of trees $t_0,t_1,t_2,...$ such that $t_i \sqsubseteq t_{i+1}$ for all $i$, one can prove that the set of all $t_i$ has a supremum that is called the \defin{limit tree} of the sequence.\medskip

A \defin{higher order recursion scheme (HORS)} $G = \langle \variables, \termi , \nontermi , \rewrite, S \rangle $ is a tuple such that: $\variables$ is a finite set of typed symbols called \defin{variables}; $\termi$ is a finite set of typed symbols of order at most 1, called the \defin{set of terminals}; $\nontermi$ is a finite set of typed symbols called \defin{set of non-terminals}; $\rewrite$ is a set of \defin{rewrite rules}, one per non terminal $F: \tau_1\rightarrow ... \rightarrow \tau_k \rightarrow o \in \nontermi$, of the form  $F \ x_1\  ... \ x_k \ \rightarrow \ e$ with $e:o \in \termes(\termi \uplus\nontermi \uplus \{x_1,...,x_k\})$; $S \in \nontermi$ is  the \defin{initial non-terminal}.

 We define the \defin{rewriting relation $\rightarrow_G$}  $\in \ \termes (\termi \uplus \nontermi)^2$ (or just $\rightarrow$ when $G$ is clear) as $t \rightarrow_G t'$ iff there exists a context $C[\bullet]$, a rewrite rule $ F\ x_1 ... x_k \rightarrow e$, and a term $F \ t_1 \ ... \ t_k:o$ such that $t = C[ F \ t_1 ... t_k ]$ and $t'= C [ e_{[x_1 \mapsto t_1]...[x_k \mapsto t_k]}]$. We call $F \ t_1 \ ... \ t_k:o$ a \defin{redex}. Finally we define $\rightarrow_G^*$ as the reflexive and transitive closure of $\rightarrow_G$. 

We define inductively the \defin{$\bot$-transformation} $(\cdot)^\bot  \ : \ \termes^o(\nontermi \uplus \termi) \rightarrow  \termes^o( \termi \uplus \{ \bot : o \})$:
$( F \ t_1 \ ...\  t_k )^\bot  \ = \ \bot \ \forall F \in \nontermi$ and $(a \ t_1 \ ... \ t_k)^\bot  \ = \  a\  t_1^\bot ... t_k^\bot $ for all $a \in \termi$. We define a \defin{derivation}, as a possibly  infinite sequence of terms linked by the rewrite relation. Let  $t_0=S \rightarrow_G t_1 \rightarrow_G t_2 \rightarrow_G ...$ be a derivation, then one can check that $(t_0)^\bot \sqsubseteq (t_1)^\bot \sqsubseteq (t_2)^\bot \sqsubseteq ...$, hence it admits a limit. One can prove that the set of all such limit trees has a greatest element that we denote $\valtree{G}$ and refer to as the \defin{value tree} of $G$. Note that $\valtree{G}$ is the supremum of $\{ t^\bot \ | \ S \rightarrow^* t \}$. Given a term $t:o$, we denote by $G_t$ the scheme obtained by transforming $G$ such that  it starts derivations with the term $t$, formally, $G_t = \langle \variables,\termi,\nontermi\uplus\{S'\}, \rewrite\uplus\{S' \rightarrow t\},S'\rangle$. One can prove that if $t \rightarrow t'$ then $\valtree{G_t}=\valtree{G_{t'}}$.\medskip

\noindent \textbf{Example}. Let $G = \langle \variables, \termi , \nontermi , \rewrite, S \rangle $ be the scheme such that: $\variables = \{ x:o, \phi:o\rightarrow o, \psi:(o\rightarrow o) \rightarrow o \rightarrow o \}$, $\termi = \{ \a : o^3 \rightarrow o, \b : o \rightarrow o \rightarrow o, \c: o\}$, $\nontermi = \{F : \big ( (o \rightarrow o) \rightarrow o \rightarrow o \big ) \rightarrow (o\rightarrow o) \rightarrow o \rightarrow o, H : (o \rightarrow o) \rightarrow o \rightarrow o, I,J,K: o \rightarrow o, S : o\}$, and $\rewrite$ contains the following rewrite rules:
\begin{equation*}
\begin{array}{lclclclclcl}
F  \ \psi \ \phi \ x & \rightarrow & \psi \ \phi \ x & \quad &
I \ x & \rightarrow & x & \quad &
H \ \phi \ x & \rightarrow  & \a \ (J \ x) \ (K \ x) \ (\phi \ x) \\
J \ x & \rightarrow & \b \ (J\ x) \ (J\ x) & \quad &
K \ x & \rightarrow & K \ (K \ x) & \quad &
S & \rightarrow  & F \ H \ I  \ \c 
\end{array}
  \end{equation*}

Here is an example of finite derivation:
\begin{multline*}
S  \quad \rightarrow \quad  F \ H \ I  \ \c \quad
  \rightarrow \quad  H \ I \ \c \quad
  \rightarrow \quad \a \ (J \ \c) \ (K \ \c) \ (I \ \c) \\
   \rightarrow \quad \a \ (J \ \c) \ (K \ (K \ \c)) \ (I \ \c) \quad
     \rightarrow  \quad \a \ (J \ \c) \ (K \ (K \ (K \ \c))) \ (I \ \c) 
\end{multline*}
If one extends it by always rewriting a redex of head $K$, its limit is the tree $\a \ \bot \ \bot \ \bot$, but this is not the value tree of $G$. The value tree $\valtree{G}$ is depicted below.

\begin{center}
\begin{tikzpicture}
\begin{scope}[scale = 0.6]
\node (1) at (0.5,0.5) {...};
\node (2) at (1,2) {\b};
\node (3) at (1.5,0.5) {...};
\node (4) at (2,3) {\b};
\node (5) at (2.5,0.5) {...};
\node (6) at (3,2) {\b};
\node (7) at (3.5,0.5) {...};
\node (8) at (4,4) {\b};
\node (9) at (4.5,0.5) {...};
\node (10) at (5,2) {\b};
\node (11) at (5.5,0.5) {...};
\node (12) at (6,3) {\b};
\node (13) at (6.5,0.5) {...};
\node (14) at (6.5,5) {\a};
\node (15) at (7,2) {\b};
\node (16) at (7.5,0.5) {...};
\node (17) at (8,4) {$\bot$};
\node (18) at (10,4) {\c};

\draw (14) -- (8);
\draw (14) -- (17);
\draw (14) -- (18);

\draw (8) -- (4);
\draw (8) -- (12);

\draw (4) -- (2);
\draw (4) -- (6);
\draw (12) -- (10);
\draw (12) -- (15);

\draw (2) -- (1);
\draw (2) -- (3);
\draw (6) -- (5);
\draw (6) -- (7);
\draw (10) -- (9);
\draw (10) -- (11);
\draw (15) -- (13);
\draw (15) -- (16);
\end{scope}
\end{tikzpicture}
\end{center}

\subsection*{Evaluation Policies}

We  now put constraints on the derivations we allow. If there are no constraints, then we say that the derivations are unrestricted and we let $\Acc^G = \{ t:o \ |\  S \rightarrow^* t\}$ be the set of accessible terms using unrestricted derivations. Given a rewriting $t \rightarrow t'$ such that $t = C[F\ s_1 \ ... \ s_k ] $ and $t' = C[ \substi{\forall j \ x_j}{s_j}{e} ]$ with $F\ x_1  ...  x_k   \rightarrow  e$$\in \rewrite$.
\begin{itemize}
\item We say that $t \rightarrow t'$ is an \defin{outermost-innermost} (OI) rewriting (written $t \rightarrow_{OI} t'$) there is no \redex containing the occurrence of $\bullet$ as a subterm of $C[\bullet]$.
\item We say that $t \rightarrow t'$ is an \defin{innermost-outermost} (IO) rewriting (written $t \rightarrow_{IO} t'$), if for all $j$ there is no \redex as a subterm of $s_j$.
\end{itemize}

Let $\Acc^G_{OI} = \{ t:o \ |\  S \rightarrow^*_{OI} t\}$ be the set of accessible terms using OI derivations and $\Acc^G_{IO} = \{ t:o \ |\  S \rightarrow^*_{IO} t\}$ be the set of accessible terms using IO derivations. There exists a supremum of $\Acc^G_{OI}$ (resp. $\Acc^G_{IO}$) which is the maximum of the limit trees of $OI$ derivations(resp. $IO$ derivations). We write it $\valtree{G}_{OI}$ (resp. $\valtree{G}_{IO}$). For all recursive scheme $G$, $(\Acc^G)^\bot = (\Acc^G_{OI})^\bot$, in particular $\valtree{G}_{OI} = \valtree{G}$. But $\valtree{G}_{IO} \sqsubseteq \valtree{G}$ and in general, the equality does not hold (see the example is the next section).

%% file: OI2IOShort.tex
 Fix a recursion scheme $G = \langle \variables, \termi , \nontermi , \rewrite, S \rangle $. Our goal is to define another scheme $\barre{G} =  \langle \barre{\variables}, \termi , \barre{\nontermi} , \barre{\rewrite}, I \rangle $ such that $\valtree{\barre{G}}_{IO}=\valtree{G}$. The idea is to add an extra argument ($\Delta$) to each non terminal, that will be required to rewrite it (hence the types are changed). We feed this argument to the outermost non terminal, and duplicate it to subterms only if the head of the term is a terminal. Hence all derivations will be IO-derivations.

We define the \defin{$\barre{( \cdot )}$ transformation} over types by $\barre{o} = o \rightarrow o$, and $\barre{\tau_1 \rightarrow \tau_2}=\barre{\tau_1} \rightarrow \barre {\tau_2}$. In particular, if $\tau = \tau_1 \rightarrow ... \rightarrow \tau_k \rightarrow o$ then $\barre{\tau} = \barre{\tau_1} \rightarrow ... \rightarrow \barre{\tau_k} \rightarrow o \rightarrow o$. Note that for all $\tau$, $\order(\barre{\tau}) = \order(\tau)+1$.

For all $x:\tau \in \variables$ we define $\barre{x} : \barre{\tau}$ as a fresh variable. Let $ar_{max}$ be the maximum arity of terminals, we define $\eta_1,...,\eta_{arity_{max}}:o \rightarrow o$ and $\delta:o$ as fresh variables, and we let $\barre{\variables} = \{\barre{x}:\barre{\tau} \ | \ x\in \variables \} \uplus \{\eta_1,...,\eta_{ar_{max}}\}\uplus\{ \delta : o \}$. Note that $\delta$ is the only variable of type $o$.  For all $a:\tau \in \termi$ define $\barre{a} : \barre{\tau}$ as a fresh \textbf{non-terminal} and for all $F:\tau \in \nontermi$ define $\barre{F} : \barre{\tau}$ as a fresh non-terminal. Let $\barre{\nontermi} = \{\barre{a}:\barre{\tau} \ | \ a\in \termi \} \uplus \{\barre{F}:\barre{\tau} \ | \ F\in \nontermi \} \uplus \{ \Delta : o, I : o \}$. Note that $I$ and $\Delta$ are the only symbols in  $\barre{\nontermi}$ of type $o$.

Let $t : \tau \in \termes(\variables \uplus \termi \uplus \nontermi)$, we define inductively the term $\barre{t} :\barre{\tau}  \in \termes(\barre{\variables} \uplus \barre{\nontermi})$: If $t = x \in \variables$ (resp. $t = a \in \termi$, $t = F \in \nontermi$), we let $\barre{t} = \barre{x} \in \barre{\variables}$ (resp. $\barre{t} = \barre{a} \in \barre{\termi}$, $\barre{t} = \barre{F} \in \nontermi$),  if $t = t_1 \ t_2:\tau$ then $ \barre{t} = \barre{t_1} \ \barre{t_2}$.

 Let $F \ x_1 \ ... \ x_k \rightarrow e$ be a rewrite rule of $\rewrite$. We define the (valid) rule $ \barre F \ \barre {x_1} \ ... \ \barre {x_k}\ \delta\  \rightarrow \ \barre {e} \ \Delta$ in $\barre \rewrite$. Let $a \in \termi$ of arity $k$, we define the rule $\barre a \ \eta_1 \ ... \ \eta_k \ \delta \ \rightarrow \ a \ (\eta_1 \ \Delta) \ ... \  (\eta_k \ \Delta)$ in $\barre \rewrite$. We also add the rule $I \ \rightarrow \ \barre{S} \ \Delta$ to $\barre{\rewrite}$.
 Finally let $\barre G = \langle \barre{\variables}, \termi , \barre{\nontermi} , \barre{\rewrite}, I \rangle $.\medskip

\noindent\textbf{Example}. Let $G = \langle \variables, \termi , \nontermi , \rewrite, S \rangle $ be the order-1 recursion scheme with $\termi = \{ \a,\c : o \}$, $\nontermi = \{ S:o, F: o \rightarrow o \rightarrow o, H : o \rightarrow o\}$, $\variables = \{x,y:o\}$, and the following rewrite rules:
\[ \begin{array}{ccccccccccc}
	S& \rightarrow & F \ (H \ \a) \ \c &\quad &
	F \ x \ y  & \rightarrow & y &\quad &
	H \ x & \rightarrow &  H \ (H \ x) 
	\end{array} \]
Then we have $\valtree{G}_{OI}=\c$ while $\valtree{G}_{IO} = \bot$ (indeed, the only IO derivation is the following $S \rightarrow F \ (H  a) \ \c \rightarrow F \ (H \ (H \ a)) \ \c \rightarrow F \ (H \ (H \ (H \ a))) \ \c \rightarrow ...$). The order-2 recursion scheme  $\barre G = \langle \barre \variables,  \termi , \barre \nontermi , \barre \rewrite, I \rangle $ is given by $\barre \nontermi = \{ I,\Delta : o , \barre S, \barre \a,\barre \c :o \rightarrow o, \barre F: (o \rightarrow o) \rightarrow (o \rightarrow o) \rightarrow o \rightarrow o, \barre H : (o \rightarrow o) \rightarrow o \rightarrow o\}$,$\barre \variables = \{\delta : o, \barre x,\barre y:o \rightarrow o\}$ and the following rewrite rules:
\[	\begin{array}{l c l c l c l c l c l}
	I & \rightarrow & \barre S \  \Delta & \quad &
	\barre S\ \delta & \rightarrow & \barre F \ (\barre H \ \barre a) \  \barre c \ \Delta & \quad &
	\barre F \ \barre x \ \barre y \ \delta & \rightarrow & \barre y \ \Delta\\
	\barre H \ \barre x \ \delta & \rightarrow &  \barre H \ (\barre H \ \barre x) \ \Delta & \quad &
	\barre \c \ \delta & \rightarrow & \c & \quad &
	\barre \a \ \delta & \rightarrow &  \a \\
	\end{array}	\]
Note that in the term $\barre F \ (\barre H \ \barre \a) \  \barre \c \ \Delta$, the subterm $\barre H \ \barre \a$ is no longer a redex since it lacks its last argument, hence it cannot be rewritten, then the only IO derivation, which is the only unrestricted derivation is $ I \rightarrow \barre S \ \Delta \rightarrow \barre F \ (\barre H \ \barre \a) \  \barre \c \ \Delta \rightarrow \barre \c \ \Delta \rightarrow \c$. Therefore $\valtree{\bar G}_{IO} = \valtree{\bar G}=\c=\valtree{G}$.


\begin{lemma}
Any derivation of $\barre G$ is in fact an OI \textbf{and} an IO derivation. Hence that $\valtree{\barre G}_{IO}= \valtree{\barre G}$.\label{barreGIOOI}\end{lemma}\begin{proof}
[Proof (Sketch)] The main idea is that the only redexes will be those that have $\Delta$ as last argument of the head non-terminal. The scheme is constructed so that $\Delta$ remains only on the outermost non-terminals, that is why any derivation is an OI derivation. Furthermore, we have that if $t= \barre F \ t_1 ... t_k \Delta$ is a redex, then none of the $t_i$ contains $\Delta$, therefore they do not contain any redex, hence $t$ is an innermost redex.\end{proof}

Note that $OI$ derivations in $\barre G$ acts like $OI$ derivations in $G$, hence $\valtree{G} = \valtree{\barre G}$.
\begin{theorem}[OI vs IO]
Let $G$ be an order-$n$ scheme. Then one can construct an order-$(n+1)$ scheme $\barre G$ such that $\valtree{G} = \valtree{\barre G}_{IO}$.\label{OI vs IO}
\end{theorem}

%% file: IntroIO2OI.tex
The goal of this section is to transform the scheme $G$ into a scheme $G''$ such that $\valtree{G''}=\valtree{G}_{IO}$. The main difference between IO and OI derivations is that some redex would lead to $\bot$ in IO derivation while OI derivations could be more productive. For example take $F:o \rightarrow o$ such that $F \ x \rightarrow c$, and $H:o$ such that $H \rightarrow a\ H$, with $a:o \rightarrow o$ and $c:o$ being terminal symbols. The term $F \ H$ has a unique $OI$ derivation, $F \ H \rightarrow_{OI} c$, it is finite and it leads to the value tree assiocated. On the other hand, the (unique) IO derivation is the following $F\ H \rightarrow F (a \ H) \rightarrow F \ (a \ (a \ H)) \rightarrow ...$ which leads to the tree $\bot$.  

The idea of the transformation is to compute a tool (based on a type system) that decides if a redex would produce $\bot$ with $IO$ derivations (Section~\ref{koba}); then we embed it into $G$ and force any such redex to produce $\bot$ even with unrestricted derivations (Section~\ref{selfcorrecting}).

%% file: kobaIOShort.tex
\subsection{The Type System}\label{koba}
Given a term $t:\tau\in \termes(\termi\uplus\nontermi)$, we define the two following properties on $t$: $\mathcal{P}_\bot(t)=$``The term $t$ has type $o$ and its associated IO valuation tree is $\bot$'', and $\mathcal{P}_\infty(t)=$``the term $t$ has not necessarily ground type, it contains a redex $r$ such that any IO derivation from $r$ producing it's IO valuation tree is infinite''. Note that $\mathcal{P}_\infty(t)$ is equivalent to ``the term $t$ contains a redex $r$ such that $\valtree{G_r}_{IO}$ is either infinite or contains $\bot$''. In this section we describe a type system, inspired from the work of Kobayashi~\cite{Kobayashi09}, that characterises if a term verifies these properties.\medskip

Let $Q$ be the set $\{q_{\bot},q_{\infty}\}$. Given a type $\tau$, we define inductively the sets $(\tau)^{atom}$ and $(\tau)^{\wedge}$ called respectively set of atomic mappings and set of conjunctive mappings:

\noindent$(o)^{atom}=Q \ , \quad (o)^{\wedge}=\{\bigwedge\{\theta_1,...,\theta_i\} \ | \ \theta_1,...,\theta_i \in Q\}\ , \quad  
(\tau_1 \rightarrow \tau_2)^{atom}= \{q_\infty\} \uplus \{(\tau_1)^\wedge \rightarrow (\tau_2)^{atom}\}$
$(\tau_1 \rightarrow \tau_2)^{\wedge}=\{\bigwedge\{\theta_1,...,\theta_i\} \ | \ \theta_1,...,\theta_i \in (\tau_1 \rightarrow \tau_2)^{atom}\}$.

We will usually use the letter $\theta$ to represents atomic mappings, and the letter $\sigma$ to represent conjunctive mappings. Given a conjunctive mapping $\sigma$ (resp. an atomic mapping $\theta$) and a type $\tau$, we write $\sigma :: \tau$ (resp. $\theta::_a \tau$) the relation $\sigma \in (\tau)^\wedge$ (resp. $\theta \in (\tau)^{atom})$. For the sake of simplicity, we identify the atomic mapping $\theta$ with the conjunctive mapping $\bigwedge\{ \theta \}$.


Given a term $t$ and a conjunctive mapping $\sigma$, we define a judgment as a tuple $\Theta \vdash t \triangleright\sigma$, pronounce ``from the environment $\Theta$, one can prove that $t$ matches the conjunctive mapping $\sigma$'', where the environment $\Theta$ is a partial mapping from $ \variables \uplus \nontermi$ to conjunctive mapping. Given an environment $\Theta$,  $\alpha\in \variables \uplus \nontermi$ and a conjunctive mapping $\sigma$, we define the environment $\Theta'=\Theta,\alpha\triangleright \sigma$ as $Dom(\Theta')=Dom(\Theta)\cup \{\alpha\}$ and $\Theta'(\alpha)= \sigma$ if $\alpha\not\in Dom(\Theta)$, $\Theta'(\alpha)= \sigma \wedge \Theta(\alpha)$ otherwise, and $\Theta'(\beta)=\Theta(\beta)$ if $\beta \neq \alpha$.

 We define the following judgement rules:

\[ \frac{\Theta \vdash t\triangleright \theta_1 \quad ... \quad \Theta \vdash t\triangleright \theta_n}
{\Theta \vdash t\triangleright \bigwedge\{ \theta_1,...,\theta_n\}} (Set)\qquad \frac{}{\Theta ,\alpha\triangleright\bigwedge\{ \theta_1,...,\theta_n\}\vdash \alpha \triangleright \theta_i} (At) \ \textit{(for all $i$)}\]

\[\frac{}{\Theta \vdash a \triangleright \sigma_1\rightarrow ... \rightarrow \sigma_{i\leq arity(a)} \rightarrow q_\infty }(\termi)\ (\textit{ for $a \in \termi$ and  $\exists j \ \sigma_j =q_\infty$ }) \]

\[\frac {\Theta \vdash t_1 \triangleright \sigma \rightarrow \theta \quad \Theta \vdash t_2 \triangleright \sigma}
{\Theta \vdash t_1 \ t_2 \triangleright \theta} (App)\quad \frac {}
{\Theta \vdash t \triangleright q_\infty  \rightarrow q_\infty} (q_\infty \rightarrow q_\infty \ I) \ (\textit{if $t:\tau_1 \rightarrow \tau_2$})\quad \frac {\Theta \vdash t_1 \triangleright q_\infty }
{\Theta \vdash t_1 \ t_2 \triangleright q_\infty}  (q_\infty \ I)\]

Remark that there is no rules that directly involves $q_\bot$, but it does not mean that no term matches $q_\bot$, since it can appear in $\Theta$. Rules like $(At)$ or $(App)$ may be used to state that a term matches $q_\bot$.

We say that $(G,t)$ matches the conjunctive mapping $\sigma$ written $\vdash (G,t) \triangleright \sigma$ if there exists an environment $\Theta$, called a witness environment of $\vdash (G,t) \triangleright \sigma$,  such that (1) $Dom(\Theta)=\nontermi$, (2) $\forall F:\tau \in \nontermi \ \Theta(F)::\tau$, (3) if $F \ x_1 ... x_k \rightarrow e$$\in \rewrite$ and $\Theta \vdash F\triangleright \sigma_1 \rightarrow ... \rightarrow \sigma_{i\leq k} \rightarrow q$ then either there exists $j$ such that $ q_\infty \in \sigma_j$, or $i=k$ and $\Theta,x_1\triangleright \sigma_1 , ... , x_k \triangleright \sigma_k \vdash e\triangleright q$, (4) $\Theta \vdash t\triangleright\sigma$.

The following two results state that this type system matches the properties $\mathcal{P}_\bot$ and $\mathcal{P}_\infty$ and furthermore we can construct a universal environment, $\Theta^\star$, that can correctly judge any term.

\begin{theorem}[Soundness and Completeness]
Let $G$ be an HORS, and $t$ be term (of any type), $\vdash(G,t) \triangleright q_\infty$ (resp. $\vdash(G,t) \triangleright q_\bot$) if and only if $\mathcal{P}_\infty(t)$ (resp. $\mathcal{P}_\bot(t)$) holds.\label{Soundness and Completeness}
\end{theorem}

\begin{proposition}[Universal Witness]
There exists an environment $\Theta^\star$ such that for all term $t$, the judgment $\vdash (G,t)\triangleright \sigma $ holds if and only if  $\Theta^\star \vdash t \triangleright \sigma$.\label{Universal Witness Short}
\end{proposition}

\begin{proof}[Proof (Sketch)]
To compute $\Theta^\star$, we start with an environment $\Theta_0$ satisfying Properties (1) and (2)  (~$Dom(\Theta_0)=\nontermi$ and $\forall F:\tau \in \nontermi \ \Theta_0(F)::\tau$~) that is able to judge any term $t:\tau$ with any conjunctive mapping $\sigma::\tau$.

 Then let $\mathcal{F}$ be the mapping from the set of environments to itself, such that for all $F:\tau_1 \rightarrow ... \rightarrow \tau_k \rightarrow o \in \nontermi$, if $F \ x_1...x_k\rightarrow e$$\in\rewrite$ then, 
\begin{multline*}\mathcal{F}(\Theta)(F) = \{\sigma_1\rightarrow ... \rightarrow \sigma_k \rightarrow q \ | \ q\in Q \wedge 
\forall i \ \sigma_i :: \tau_i \wedge 
\Theta,x_1\triangleright \sigma_1,...,x_k\triangleright \sigma_k \vdash e : q\}  \\
\cup\{\sigma_1 \rightarrow ... \rightarrow \sigma_{i\leq k} \rightarrow q_\infty \ | \  \wedge \forall i \ \sigma_i :: \tau_i \wedge \exists j \ q_\infty \in \sigma_j \}\\
\cup\{\sigma_1 \rightarrow ... \rightarrow \sigma_{ k} \rightarrow q_\bot \ | \ \forall i \ \sigma_i :: \tau_i \wedge \exists j \ q_\infty \in \sigma_j \}.\end{multline*}

We iterate $\mathcal{F}$ until we reach a fixpoint. The environment we get is $\Theta^\star$, it verifies properties (1) (2) and (3). Furthermore we can show that this is the maximum of all environment satisfying these properties, i.e. if $\vdash (G,t)\triangleright \sigma $ then   $\Theta^\star \vdash t \triangleright \sigma$.
\end{proof}

%% file: selfcorrectingShort.tex
\subsection{Self-Correcting Scheme}\label{selfcorrecting}

For all term $t:\tau \in \termes(\termi\uplus\nontermi)$,  we define $\sem{t} \in (\tau)^\wedge $, called the semantics of $t$, as the conjunction of all atomic mappings $\theta$ such that $\Theta^\star \vdash t \triangleright \theta$ (recall that $\Theta^\star$ is the environment of Proposition~\ref{Universal Witness Short}). In particular $\mathcal{P}_\bot(t)$ (resp. $\mathcal{P}_\infty(t)$) holds if and only if $q_\bot \in \sem{t}$ (resp. $q_\infty \in \sem{t}$). Given two terms $t_1 : \tau_2 \rightarrow \tau$ and $t_2 : \tau_2$ the only rules we can apply to judge $\Theta^\star  \vdash t_1 \ t_2 \triangleright \theta $ are $(App)$, $(q_\infty \rightarrow q_\infty \ I)$ and $(q_\infty \ I)$. We see that $\theta$ only depends on which atomic mappings are matched by $t_1$ and $t_2$. In other words $\sem{t_1 \ t_2}$ only depends on $\sem{t_1}$ and $\sem{t_2}$, we write $\sem{t_1} \ \sem{t_2} = \sem{t_1 \ t_2}$.\medskip

In this section, given a scheme $G=\langle \variables, \termi,\nontermi,\rewrite, S \rangle$, we transform it into $G'=\langle \variables', \termi,\nontermi',\rewrite', S \rangle$ which is basically the same scheme except that while it is producing an IO derivation, it evaluates $\sem{t'}$ for any subterm $t'$ of the current term and label $t'$ with $\sem{t'}$. Note that  if $t \rightarrow_{IO} t'$, then $\sem{t} = \sem{t'}$. Since we cannot syntactically label terms, we will label all symbols by the semantics of their arguments, \emph{e.g.} if we want to label $F \ t_1 ... t_k$, we will label $F$ with the $k$-tuple $(\sem{t_1},...,\sem{t_k})$. 

A problem may appear if some of the arguments are not fully applied, for example imagine we want to label $F \ H$ with $H:o \rightarrow o$. We will label $F$ with $\sem{H}$, but since $H$ has no argument we do not know how to label it. The problem is that we cannot wait to label it because once a non-terminal is created, the derivation does not deal explicitly with it. The solution is to create one copy of $H$ per possible semantics for its argument (here there are four of them: $\bigwedge\{ \}, \bigwedge\{ q_\bot \},\bigwedge\{ q_\infty\}, \bigwedge \{ q_\bot, q_\infty\}$). This means that $F^{\sem{H}}$ would not have the same type as $F$: $F$ has type $(o\rightarrow o) \rightarrow o$, but $F^{\sem{G}}$ will have type $(o\rightarrow o)^4 \rightarrow o$. Hence, $F \ H$ will be labelled the following way: $ F^{\sem{H}} \ H^{\bigwedge\{ \}}H^{\bigwedge\{ q_\bot \}}H^{\bigwedge\{ q_\infty\}}H^{\bigwedge\{ q_\bot,q_\infty \}} $. Note that even if $F$ has $4$ arguments, it only has to be labelled with one semantics since all four arguments represent different labelling of the same term. We now formalize these notions.\medskip

Let us generalize the notion of semantics to deals with terms containing some variables. Given an environment on the variables $\Theta^\variables$ such that $Dom(\Theta^\variables)\subseteq \variables$  and if $x:\tau$ then $\Theta^\variables(x) :: \tau$, and given a term $t:\tau \in \termes(\termi\uplus\nontermi \uplus Dom(\Theta^\variables) )$, we define $\sem{t}_{\Theta^\variables} \in (\tau)^\wedge $,  as the conjunction of all atomic mappings $\theta$ such that $\Theta^\star,\Theta^\variables \vdash t \triangleright \theta$. Given two terms $t_1 : \tau_2 \rightarrow \tau$ and $t_2 : \tau_2$ we still have that $\sem{t_1 \ t_2}_{\Theta^\variables}$ only depends on $\sem{t_1}_{\Theta^\variables}$ and $\sem{t_2}_{\Theta^\variables}$.\medskip

To  a type $\tau = \tau_1 \rightarrow ... \rightarrow \tau_k \rightarrow o$ we associate the integer 
$\nbvar{\tau} = Card ( \{ (\sigma_1,...,\sigma_k) \ | \ \forall i \ \sigma_i \in (\tau_i)^{\wedge}  \}  )$
 and a complete ordering of  $\{ (\sigma_1,...,\sigma_k) \ | \ \forall i \ \sigma_i \in (\tau_i)^{\wedge}  \}$ denoted $\vec \sigma_1^\tau$, $\vec \sigma_2^\tau$, ... , $\vec \sigma_{\nbvar{\tau}}^\tau$. We define inductively the type $\tau^+ = (\tau_1^+)^{\nbvar{\tau_1}} \rightarrow ... \rightarrow (\tau_k^+)^{\nbvar{\tau_k}} \rightarrow o$. 
 
To a non terminal $F: \tau_1 \rightarrow ... \rightarrow \tau_k \rightarrow o$ (resp. a variable $x:  \tau_1 \rightarrow ... \rightarrow \tau_k \rightarrow o$) and a tuple $\sigma_1::\tau_1,...,\sigma_k::\tau_k$, we associate the non-terminal $F^{\sigma_1,...,\sigma_k} : \tau_1^{\nbvar{\tau_1}}\rightarrow ... \rightarrow \tau_k^{\nbvar{\tau_k}} \rightarrow o \in \nontermi'$ (resp. a variable $x^{\sigma_1,...,\sigma_k} : \tau_1^{\nbvar{\tau_1}}\rightarrow ... \rightarrow \tau_k^{\nbvar{\tau_k}} \rightarrow o\in \variables'$). 

Given a term $t:\tau = \tau_1 \rightarrow ... \rightarrow \tau_k \rightarrow o \in \termes(\variables \uplus \termi \uplus \nontermi)$ and an environment on the variables $\Theta^\variables$ such that $Dom(\Theta^\variables)\subseteq \variables$ contains all variables in $t$,  we define inductively the term $t_{\Theta^\variables}^{+\sigma_1,...,\sigma_k}:\tau^+ \in \termes(\variables' \uplus \termi' \uplus \nontermi')$ for all $\sigma_1::\tau_1,...,\sigma_k::\tau_k$. If $t=F\in \nontermi$ (resp. $t=x\in \variables$), $t_{\Theta^\variables}^{ +  \sigma_1,...,\sigma_k}=F^{\sigma_1,...,\sigma_k}$ (resp. $t_{\Theta^\variables}^{+\sigma_1,...,\sigma_k} = x^{\sigma_1,...,\sigma_k}$), if $t = a \in \termi$, $t_{\Theta^\variables}^{+ \sigma_1,...,\sigma_k}=a$. Finally consider the case where $t = t_1 \ t_2$ with $t_1 : \tau ' \rightarrow \tau$ and $t_2 : \tau '$. Let $\sigma=\sem{t_2}_{\Theta^\variables}$. Remark that ${t_1}_{\Theta^\variables}^{+ \sigma,\sigma_1,...,\sigma_k}: (\tau'^+)^{\nbvar{\tau'}} \rightarrow \tau^+$. We define $(t_1 \ t_2)_{\Theta^\variables}^{+ \sigma_1,...,\sigma_k} = {t_1}_{\Theta^\variables}^{+ \sigma,\sigma_1,...,\sigma_k} \ {t_2}_{\Theta^\variables}^{+ \vec \sigma^{\tau'}_1} ... \ {t_2}_{\Theta^\variables}^{+ \vec \sigma^{\tau'}_{\nbvar{\tau'}}}$. Note that since this transformation is only duplicating and anotating, given a term $t^{+\sigma_1,...,\sigma_k}$ we can uniquely find the unique term $t$ associated to it.

Let $F:\tau_1 \rightarrow ... \rightarrow \tau_k \rightarrow o \in \nontermi$, $\sigma_1::\tau_1,...,\sigma_k::\tau_k$, and $\Theta^\variables =x_1 \triangleright \sigma_1,...,x_k \triangleright \sigma_k$ . If $F \ x_1 ... x_k \ \rightarrow \ e$$ \in \rewrite$, we define in $\rewrite'$ the rule $\ F^{\sigma_1,...,\sigma_k} \ x_1^{+\vec \sigma^{\tau_1}_1} ... \ x_1^{+ \vec \sigma^{\tau_1}_{\nbvar{\tau_1}}} \ ... \ \ x_k^{+ \vec \sigma^{\tau_k}_1} ... \ x_k^{+\vec \sigma^{\tau_k}_{\nbvar{\tau_k}}} \ \rightarrow \ e^+_{\Theta^\variables}$. Finally, recall that $G'=\langle \variables', \termi,\nontermi',\rewrite', S \rangle$.\medskip

The following theorem states that $G'$ is just a labeling version of $G$ and that it acts the same.

\begin{theorem}[Equivalence between $G$ and $G'$]
Given a term $t:o$, $\valtree{G'_{t^+}}_{IO}=\valtree{G_t}_{IO}$.\label{Equivalence between $G$ and $G'$}
\end{theorem}

We transform $G'$ into the scheme $G''$ that will directly turn into $\bot$ a redex $t$ such that $q_\bot \in \sem{t}$. For technical reason, instead of adding $\bot$ we add a non terminal $Void:o$ and a rule $Void \rightarrow Void$. $G'=\langle \variables', \termi,\nontermi'\uplus\{Void:o\},\rewrite'', S \rangle$ such that $\rewrite''$ contains the rule $Void \rightarrow Void$ and for all $F \in \nontermi$, if $q_\bot \in \sem{F}\ \sigma_1 \ ... \ \sigma_k$ then 
$ F^{\sigma_1,...,\sigma_k} \ x_1^{+\vec \sigma^{\tau_1}_1} ... \ x_1^{+ \vec \sigma^{\tau_1}_{\nbvar{\tau_1}}}  ... \ x_k^{+ \vec \sigma^{\tau_k}_1} ... \ x_k^{+\vec \sigma^{\tau_k}_{\nbvar{\tau_k}}}  \rightarrow  Void$ otherwise we keep the rule of $\rewrite'$.\medskip

  The following theorem concludes Section~\ref{IO2OI}.
 
 \begin{theorem}[IO vs OI] Let $G$ be a higher-order recursion scheme. Then one can construct a scheme $G''$ having the same order of $G$ such that $\valtree{G''}=\valtree{G}_{IO}$.
 \end{theorem}
 
 \begin{proof}[Proof (Sketch)]
  First, given a term $t:o$, one can prove that $\valtree{G''_{t^+}}_{IO} = \valtree{G'_{t^+}}_{IO}$. 
  
  Then take a redex $t$ such that $\valtree{G''_t}_{IO}= \bot$, i.e. $q_\bot \in \sem{G_t}$. There is only one OI derivation from $t$: $t \rightarrow Void \rightarrow Void \rightarrow...$, then $\valtree{G''_t}= \bot$. We can extend this result saying that if there is the symbol $\bot$ at node $u$ in $\valtree{G''_t}_{IO}$, then there is $\bot$ at node $u$ in $\valtree{G''_t}$. Hence, since $\valtree{G''_t}_{IO} \sqsubseteq \valtree{G''_t}$, we have $ \valtree{G''} =\valtree{G''}_{IO} $. Then $\valtree{G''}=\valtree{G''}_{IO} = \valtree{G'}_{IO} =\valtree{G}_{IO}$.
 
 \end{proof}

%% file: Conclusion.tex
We have shown that value trees obtained from schemes using innermost-outermost derivations (IO) are the same as those obtained using unrestricted derivations. More precisely, given an order-$n$ scheme $G$ we create an order-$(n+1)$ scheme $\barre G$ such that $\valtree{\barre G}_{IO} = \valtree{G}$. However, the increase of the order seems unavoidable. We also create an order-$n$ scheme $G''$ such that $\valtree{\barre G''} = \valtree{G}_{IO}$. In this case the order does not increase, however the size of the scheme deeply increases while it remains almost the same in $\barre G$.